\documentclass[12pt]{amsart}
\usepackage[utf8]{inputenc}
\usepackage[T1]{fontenc}
\usepackage{amssymb}
 \usepackage{amsaddr}
\usepackage{graphics}
\usepackage{xcolor}
\usepackage{enumerate}
\usepackage[pagebackref, colorlinks = true, linkcolor = teal, urlcolor  = teal, citecolor = violet]{hyperref}
\usepackage[margin=1in]{geometry}
\usepackage{cleveref}
\crefname{equation}{Eq.}{Eqs.}
\usepackage{bbm}

\newtheorem{theorem}{Theorem}
\newtheorem*{theorem*}{Theorem}

\newtheorem{corollary}[theorem]{Corollary}
\newtheorem*{corollary*}{Corollary}

\newtheorem{remark}[theorem]{Remark}
\newtheorem{example}[theorem]{Example}

\newtheorem{condition}{Condition}

\newcommand{\id}{{\rm Id}}

\newcommand{\tr}{{\operatorname{tr}}}

\begin{document}

\title{Commuting Kraus operators are normal}

\author{Martin Fraas}
\address{Department of Mathematics, University of California, Davis, CA 95616, USA}
\email{mfraas@ucdavis.edu}

\begin{abstract}
Let $\{V_1, \dots, V_n \}$ be a set of mutually commuting matrices. We show that if $V_1^* V_1 + \cdots +V_n^* V_n = \id$  then the matrices are normal and, in particular, simultaneously diagonalizable.
\end{abstract}

\date{\today}

\maketitle

\section{Result}

Let $\mathcal{H}$ be a finite dimensional complex Hilbert space. A matrix $A$ is called normal if it commutes with its Hermitian conjugate matrix $A^*$. A set of mutually commuting normal matrices is simultaneously diagonalizable, and there exists an orthonormal basis whose elements are joint eigenvectors of all the matrices. The condition that the matrices are normal is necessary, if two matrices $A,B$ commute, $[A,B]=0$, then without further assumptions they only need to have one common eigenvector \cite{wikic}. We provide a criteria that implies this condition.

\begin{theorem}
\label{thm}
Let $\{V_1, \dots, V_n\}$ be a set of commuting matrices such that $\sum_{\alpha=1}^n V_\alpha^* V_\alpha =\id$. Then the matrices are normal and, in particular, simultaneously diagonalizable.
\end{theorem}

The motivation for the theorem comes from the theory of quantum trajectories. In this context, a set of matrices satisfying the normalization condition
\begin{equation}
\label{eq:normalization}
\sum_{\alpha=1}^n V_\alpha^* V_\alpha = \id
\end{equation}
are called Kraus operators or jump operators. The operators describe evolution of an initial state $\psi \in \mathcal{H}$ conditioned on outcomes $\alpha_1, \dots, \alpha_n$ of a repeated measurement, and the probability of obtaining these outcomes. The probability is given by
\begin{equation}
\label{eq:P}
\mathbb{P}(\alpha_1, \dots, \alpha_n) := \| V_{\alpha_n} \dots V_{\alpha_1} \psi \|^2,
\end{equation}
which is a probability measure provided that $\|\psi \| = 1$. Following the experiments of the quantum optics group of S.~Haroche \cite{Haroche}, the case of non-demolition measurements attracted a lot of attention in physics and math literature \cite{BaBe, BaBeBe1, BaBeBe2, BFFS}. A standard definition of a non-demolition measurement postulates which observable $N$ is not being demolished.
\begin{condition}\label{defA} A set of jump operators $V_\alpha$ is called non-demolition if there exists a Hermitian operator $N$ and complex functions $f_\alpha$ such that $V_\alpha = f_\alpha(N)$.
\end{condition}
In particular, in the non-demolition case the jump operators commute, and the probability measure (\ref{eq:P}) is exchangeable. This suggest a natural intrinsic definition of the non-demolition case.
\begin{condition}\label{defB} A set of jump operators $V_\alpha$ is called non-demolition if the operators $V_\alpha$ are mutually commuting.
\end{condition}
An immediate corollary of Theorem~\ref{thm} is that these two definitions give the same class of Kraus operators.
\begin{corollary}
Conditions \ref{defA} and \ref{defB} are equivalent.
\end{corollary}
\begin{proof}
As discussed above Condition \ref{defA} implies Condition \ref{defB}. In the opposite direction, if Kraus operators $V_\alpha$ are mutually commuting then by Theorem~\ref{thm} they are normal and hence there exists a Hermitian matrix $N$ such that all matrices $V_\alpha$ are functions of $N$. 
\end{proof}

\section{Completely positive maps}
The proof of the theorem that we give below involves some results about completely positive maps. We recall these results.  A completely positive map $\Phi : B(\mathcal{H}) \to B(\mathcal{H})$ has a form
$$
\Phi(X) = \sum_\alpha V_\alpha^* X V_\alpha,
$$
for some finite set of operators $V_\alpha$ (that do not necessarily commute). We will also use the dual completely positive map
$$
\Phi^*(A) = \sum_\alpha V_\alpha A V_\alpha^*.
$$

 The first fact, see \cite{wolftour}, about completely positive maps that we will need is that 
\begin{equation}
\label{eq:norm}
\| \Phi \| = \| \sum_\alpha V_\alpha^* V_\alpha \|,
\end{equation}
where both norms are the operator norms. If $\Phi$ satisfies the normalization condition Eq.~(\ref{eq:normalization}) then $\|\Phi\|=1$ belongs to the spectrum of $\Phi$. Furthermore, there exists an positive eigen-matrix $\rho$ such that $\Phi^* (\rho) = \rho$. We will call such eigen-matrix a stationary state. 

We now assume the normalization condition Eq.~(\ref{eq:normalization}). The second fact, see \cite{CT2015} for items (i), (ii) and \cite[Proposition 3]{Albert2019} for (iii), that we will use is that there exists a decomposition of the Hilbert space $\mathcal{H} = \mathcal{H}_F \oplus \mathcal{H}_D$ such that with respect to this decomposition the matrices have a form
\begin{equation}
\label{eq:dec}
V_\alpha = \left( \begin{array}{cc}
				A_\alpha & B_\alpha \\
				0 & C_\alpha
			\end{array} \right),
\end{equation}
and 
\begin{enumerate}[(i)]
\item $\Phi^*$ restricted to $B(\mathcal{H}_F)$,
$$
\Phi^*_F (X) = \sum_\alpha A_\alpha X A_\alpha^*,
$$
possesses a faithful (full rank) stationary state,
\item $\Phi$ restricted to $B(\mathcal{H}_D)$,
$$
\Phi_D (X) = \sum_\alpha C_\alpha^* X C_\alpha,
$$
has a spectral radius strictly less than 1, 

\item Any solution of $\Phi(X) = X$ is diagonal with respect to the decomposition, i.e. it has a form
$$
X = \left( \begin{array}{cc}
				X_F & 0 \\
				0 & X_D
			\end{array} \right).
$$
\end{enumerate}
If the decomposition is trivial, $\mathcal{H}_D = 0$, we define $B_\alpha = C_\alpha =0$. The decomposition is trivial  if and only if $B_\alpha \equiv 0$.

We will not need the full claim (iii). We will use a claim that any solution of $\Phi(X) = X$ that acts as zero on $\mathcal{H}_F$ is a zero matrix. This also follows from (ii). By (ii), subspace $\mathcal{H}_D$ is transient, i.e. for any matrix $A$,
$$
\lim_{n \to \infty} P_D {\Phi^*}^n(A) P_D = 0,
$$
where $P_D$ is the projection on $\mathcal{H}_D$. By positivity of $\Phi$, also $P_D {\Phi^*}^n(A)$ and ${\Phi^*}^n(A) P_D$ go to zero. Since $\tr(XA) = \tr(X {\Phi^*}^n (A))$ this implies that $\tr(X A) = 0$ so $X$ is indeed zero.

\section{Proof of the theorem}
We use the decomposition Eq.~(\ref{eq:dec}). The normalization Eq.~(\ref{eq:normalization}) is equivalent to
\begin{align}
\sum_{\alpha} A_\alpha^* A_\alpha &= \id, \label{n1} \\
\sum_\alpha  A_\alpha^* B_\alpha &= 0 ,\label{n2} \\
\sum_\alpha B_\alpha^* B_\alpha + C_\alpha^* C_\alpha &= \id. \label{n3}
\end{align}
The commutation assumption is equivalent to equations
\begin{align}
A_\alpha A_\beta &=  A_\beta A_\alpha, \label{c1} \\
A_\alpha  B_\beta + B_\alpha C_\beta &= A_\beta B_\alpha  + B_\beta C_\alpha \label{c2}, \\
C_\alpha C_\beta &=C_\beta C_\alpha \label{c3},
\end{align}
holding for all $\alpha, \beta$. Eqs.(\ref{n1}), (\ref{c1}) imply that the matrices $A_1, \dots ,A_n$ are mutually commuting and normalized as in Eq.~(\ref{eq:normalization}). By the property of the decomposition,
the corresponding map $\Phi_F^*$ has a faithful (full rank) stationary solution $\rho$, i.e. a matrix such that $\Phi_F^*(\rho) = \rho$. Let $\Phi_F$ be the dual map,
$$
\Phi_F(X) = \sum_{\alpha} A_\alpha^* X A_\alpha.
$$
For any $\beta$, we then have 
$$
\Phi_F(A_\beta A_\beta^*) - A_\beta \Phi_F(A_\beta^*) -  \Phi_F(A_\beta) A_\beta^*  +A_\beta A_\beta^* = \sum_\alpha |[A^*_\beta, A_\alpha]|^2. 
$$
Since the matrices commute we have that $\Phi_F(A_\beta) = A_\beta$ and $\Phi_F(A_\beta^*) = A_\beta^*$. Hence we obtain 
$$
\sum_\alpha \tr \left( \rho|[A^*_\beta, A_\alpha]|^2\right) =0,
$$
and it follows that the matrices $A_\alpha$ are normal.

We now show that $B_\alpha=0$ which implies that $\mathcal{H}_D = 0$ and finishes the proof. Multiplying Eq.~(\ref{c2}) from the left by $A_\alpha^*$ and summing over $\alpha$ we get
$$
B_\beta = \sum_\alpha A_\alpha^* B_\beta C_\alpha,
$$
where we used that $[A_\beta^*,A_\alpha]=0$ and Eqs.~(\ref{n2}), (\ref{n1}).  The equation implies that
$$
\Phi \left( \begin{array}{cc}
				0 & B_\beta \\
				0 & Y
			\end{array} \right) = \left( \begin{array}{cc}
				0 & B_\beta \\
				0 & \Phi_D(Y) + \sum_\alpha B_\alpha^* B_\beta B_\alpha 
			\end{array} \right).
$$
Since $\Phi_D$ has spectral radius less than $1$, the sum 
$$
X_D : = \sum_{n=0}^\infty \Phi_D^n(\sum_\alpha B_\alpha^* B_\beta B_\alpha)
$$
is convergent. Hence the matrix
$$
X = \left( \begin{array}{cc}
				0 & B_\beta \\
				0 & X_D
			\end{array} \right)
$$
satisfies the equation $\Phi X = X$. Any solution of the equation has to be diagonal so we got $B_\beta = 0$ as announced. This finishes the proof.

\begin{remark}
\label{remark}
A related result is given in \cite{BFFS}. Assuming that $V_\alpha$'s are commuting and that the decomposition Eq.~(\ref{eq:dec}) is trivial, \cite{BFFS} shows that the quantum trajectory purifies on the joint spectral decomposition of $A_\alpha^* A_\alpha$. For the notion of purification on the spectrum we refer the reader to the aforementioned article, and only note that this, in particular, means that the spectral measure of $\psi$ associated to the joint spectral decomposition gives the de Finetti decomposition of the measure Eq.~(\ref{eq:P}).  For an example, see Example~\ref{example}.
\end{remark}

\section{Discussion}
The assumption that the Hilbert space is finite dimensional cannot be relaxed. Any partial isometry gives a counterexample. For example, the right shift $R$ on $L^2(\mathbb{N})$,
$$
R (x_1, x_2 ,x_3, \dots) = (0, x_1, x_2, x_3, \dots), 
$$
satisfies the normalization condition $R^* R =\id$ (and obviously $[R,R]=0$)  but $R$ is not normal. The reason for the failure of the theorem is that the decomposition Eq.~(\ref{eq:dec}) is more complicated in infinite dimensions, see \cite{CarbonePautrat}, there are invariant subspaces that do not posses any stationary state.

The theory of non-demolition quantum trajectories in infinite dimensional Hilbert spaces assuming the definition given by Condition~\ref{defA} was developed in \cite{BCFF}. It would be interesting to understand the behavior of measure Eq.~(\ref{eq:P}) under Condition~\ref{defB}. We illustrate the difference on two examples.

\begin{example}
\label{example}
Consider the space $L^{2}(\mathbb{Z})$ and jump operators
$$
V_1 = \frac{1}{2}(1+R), \quad V_2 = \frac{1}{2} (1 - R),
$$
where $R$ is the right shift operator on $\mathbb{Z}$. Then $V_1, V_2$ are mutually commuting and normal. Let $\psi \in L^{2}(\mathbb{Z})$ and let $\hat{\psi} \in L^2[0,2 \pi]$ be it's Fourrier transform, i.e.
$$
\psi_n = \frac{1}{\sqrt{2 \pi}} \int_0^{2 \pi} e^{i n k} \hat{\psi}(k) dk.
$$
\end{example}
Then for the measure Eq.~(\ref{eq:P}) we get 
\begin{align*}
\mathbb{P}(\alpha_1, \dots, \alpha_n) &=  \int_0^{2 \pi} |\frac{1 + e^{-ik}}{2}|^{2 n_1} |\frac{1 - e^{-ik}}{2}|^{2 n_2} |\hat{\psi}(k)|^2 dk \\
							 &=  2^{-n}\int_0^{2 \pi} (1 + \cos k)^{n_1} (1- \cos k)^{n_2} |\hat{\psi}(k)|^2 dk
\end{align*}
where $n_1, n_2$ are the number of ones resp. twos in the sequence $\alpha_1, \dots, \alpha_n$. 
The i.i.d. measures under the integral are the same for $k$ and $\pi - k$ so that de Finetti decomposition of the measure is
$$
\mathbb{P}(\alpha_1, \dots, \alpha_n) =  2^{-n}\int_{\pi/2}^{3 \pi/2} (1 + \cos k)^{n_1} (1- \cos k)^{n_2} (|\hat{\psi}(k)|^2 + |\hat{\psi}(\pi - k)|^2) dk.
$$
The measure $(|\hat{\psi}(k)|^2 + |\hat{\psi}(\pi - k)|^2) dk$ in the decomposition is the spectral measure of $\psi$ with respect to the discrete Laplacian $R + R^*$. This demonstrates Remark~\ref{remark}, it is a general feature that assuming Condition~\ref{defA}, the measure in the de Finetti decomposition is the spectral measure associated with $N$, see \cite{BCFF}.

This example should be contrasted with
\begin{example}
Consider the space $L^{2}(\mathbb{N})$ and jump operators
$$
V_1 = \frac{1}{2}(1+R), \quad V_2 = \frac{1}{2} (1 - R),
$$
where $R$ is the right shift operator on $\mathbb{N}$. Then $V_1, V_2$ are mutually commuting but not normal. The action of $V_\alpha$ on any $\psi$ is the same as if $V_\alpha$ was defined on the whole $\mathbb{Z}$ and $\psi$ was extended by zero on the negative numbers. Hence de Finetti decomposition given in the previous example holds true also here. However, the measure is not anymore associated to the spectral measure of $R +R^*$ on $L^2(\mathbb{N})$.
\end{example}

\section*{Acknowledgements} The author thanks Tristan Benoist for comments on the draft version of the article. The research has been supported by ANR project QTraj (ANR-20-CE40-0024-01) of the French National Research Agency (ANR).

\bibliography{QTraj}
\bibliographystyle{abbrv}
\end{document}